\newif\ifIACRCC
\IACRCCtrue

\PassOptionsToPackage{hypertexnames=false}{hyperref}

\ifIACRCC
\documentclass[
  version=preprint,
  biblatex
]{iacrcc}
\else
\documentclass{article}
\fi

\usepackage[utf8]{inputenc}
\usepackage[T1]{fontenc}

\ifIACRCC
\else
\usepackage{authblk}
\usepackage[
  backend=biber,
  minalphanames=3,
  maxalphanames=4,
  maxbibnames=99,
  style=alphabetic
]{biblatex}
\fi

\usepackage{csquotes}
\usepackage[a4paper]{geometry}
\usepackage[acronym]{glossaries}
\usepackage{graphicx}
\usepackage{xcolor}
\usepackage{amsmath,amssymb}
\ifIACRCC
\else
\usepackage{hyperref}
\fi
\usepackage{upgreek}
\usepackage[english]{babel}
\usepackage{amsthm}
\usepackage{amssymb}
\usepackage{stmaryrd}
\usepackage{mathtools}

\usepackage{tikz}
\usepackage{standalone}
\usetikzlibrary{calc,fit,positioning,decorations.pathreplacing}
\tikzset{>=latex}
\usepackage{subcaption}
\usepackage{rotating}
\ifIACRCC
\license{CC-by}
\fi

\usepackage{listings}
\usepackage{xcolor}
\usepackage[noend]{algorithm}
\usepackage[noend]{algpseudocode}

\makeatletter
\patchcmd{\lst@MakeCaption}{\addcontentsline{lol}}{\addcontentsline{loa}}{}{}
\patchcmd{\lst@MakeCaption}{\addcontentsline{lol}}{\addcontentsline{loa}}{}{}
\AtBeginDocument{%
  \let\c@lstlisting\c@algorithm

}
\makeatother

\floatname{algorithm}{Listing}
\algblockdefx{oGame}{oEndGame}%
[2]{\textbf{game} #1(#2)}%
{\textbf{end game}}
\algblockdefx{oQuery}{oEndQuery}%
[2]{\textbf{oracle} #1(#2)}%
{\textbf{end oracle}}
\makeatletter
\ifthenelse{\equal{\ALG@noend}{t}}%
  {\algtext*{oEndGame}}
  {}%
\ifthenelse{\equal{\ALG@noend}{t}}%
  {\algtext*{oEndQuery}}
  {}%
\makeatother

\usepackage{booktabs}

\ifIACRCC
\else
\usepackage[
  capitalise,
  nameinlink,
  noabbrev,
]{cleveref}
\crefname{algorithm}{Listing}{Listings}  
\fi

\lstdefinestyle{sage}{
  language=Python,
  morekeywords={sage,ZZ,GF,Matrix},
  basicstyle=\ttfamily\small,
  keywordstyle=\color{blue},
  commentstyle=\color{gray},
  stringstyle=\color{orange},
  showstringspaces=false,
  frame=single,
  breaklines=true
}

\DeclareMathOperator{\MAC}{MAC}

\DeclarePairedDelimiter{\set}{\{}{\}}
\DeclarePairedDelimiter{\ind}{\llbracket}{\rrbracket}

\newcommand{\Enc}{\ensuremath{\mathsf{Encode}}}
\newcommand{\Dec}{\ensuremath{\mathsf{Decode}}}
\newcommand{\Encrypt}{\ensuremath{\mathsf{Encrypt}}}
\newcommand{\Decrypt}{\ensuremath{\mathsf{Decrypt}}}
\newcommand{\Share}{\ensuremath{\mathsf{Share}}}
\newcommand{\Rec}{\ensuremath{\mathsf{Recover}}}

\newcommand{\orwrite}[1]{{\normalfont\textsc{#1}}}

\newcommand{\oCorruptS}{\orwrite{SSS.Corrupt}}
\newcommand{\oCorruptR}{\orwrite{R.Corrupt}}
\newcommand{\oRelay}{\orwrite{Relay}}
\newcommand{\oIndSSS}{\orwrite{Ind-SSS}}
\newcommand{\oShiftSSS}{\orwrite{Shift-Robust-SSS}}
\newcommand{\oIndRelay}{\orwrite{Ind-Relay}}
\newcommand{\oForgeRelay}{\orwrite{Forge-Relay}}

\newcommand{\cs}[1]{{\textcolor{magenta}{[CS: #1]}}}

\newcommand{\getr}{\ensuremath{\overset{\$}{\gets}}}

\newcommand{\vect}[1]{\vec{#1}}
\newcommand{\qij}[3]{\ensuremath{q_{#1#2}}}
\newcommand{\cij}[3]{\ensuremath{c_{#1#2}}}
\newcommand{\cpij}[3]{\ensuremath{c'_{#1#2}}}

\newacronym{its}{ITS}{Information Theoretically Secure}
\newacronym{itsy}{ITS}{Information-Theoretic Security}
\newacronym{mac}{MAC}{Message Authentication Code}
\newacronym{otp}{OTP}{One-Time Pad}
\newacronym{qkd}{QKD}{Quantum Key Distribution}
\newacronym{vss}{VSS}{Verifiable Secret Sharing}
\newacronym{amd}{AMD}{Algebraic Manipulation Detection}
\newacronym{ake}{AKE}{Authenticated Key Exchange}

\ifIACRCC
\newtheorem{construction}{Construction}
\else
\newtheorem{theorem}{Theorem}
\newtheorem{lemma}[theorem]{Lemma}
\newtheorem{definition}[theorem]{Definition}
\newtheorem{construction}[theorem]{Construction}
\fi

\makeatletter
\newcommand{\anonurl}[1]{%
  \if@anonymous
    \texttt{[URL hidden for anonymous submission]}%
  \else
    \url{#1}%
  \fi
}
\makeatother

\ifIACRCC
\title[
  running = {Integrity from AMD in trusted repeater networks},
]{Integrity from Algebraic Manipulation Detection in Trusted-Repeater QKD Networks}
  \addauthor[
    orcid = {0009-0006-3158-6202},
    inst = {1,2},
    footnote = {Funded by the Dutch National Science Agenda (Nationale Wetenschapsagenda, NWA) by The Dutch Research Council (NWO) under project number 1436.20.002, Quantum Impact on Societal Security.},
    surname = {Robertson}
  ]{Ailsa Robertson}
  \addauthor[
    orcid = {0000-0002-1754-1415},
    inst = {1,2},
    footnote = {Partially supported by gravitation project Challenges in Cyber Security (CiCS) with file number 024.006.037 which is financed by the NWO.},
    surname = {Schaffner}
  ]{Christian Schaffner}
  \addauthor[
    orcid = {0009-0007-4379-2983},
    footnote = {Supported by the Dutch National Growth Fund (NGF), as part of the Quantum Delta NL programme.},
    inst = {1,2},
    email = {s.r.verschoor@uva.nl},
    surname = {Verschoor}
  ]{Sebastian R. Verschoor}
  \addaffiliation[ror = 04dkp9463,
    country={The Netherlands}]{University of Amsterdam}
  \addaffiliation[ror = 00zq3ce72,
    country={The Netherlands}]{QuSoft}
\keywords{QKD,
    quantum key distribution,
    secret sharing,
    authentication codes,
    trusted repeater networks,
    information theoretic security,
    provable security
}
\else
  \title{Integrity from Algebraic Manipulation Detection in Trusted-Repeater QKD Networks}
  \author[1,2]{Ailsa Robertson}
  \author[1,2]{Christian Schaffner}
  \author[1,2]{Sebastian R. Verschoor}
  \affil[1]{University of Amsterdam}
  \affil[2]{QuSoft}
\fi

\begin{document}
\maketitle
\begin{abstract}
    \Gls{qkd} allows secure communication without relying on computational assumptions,
    but can currently only be deployed over relatively short distances due to hardware constraints.
    To extend \gls{qkd} over long distances, networks of trusted repeater nodes can be used,
    wherein \gls{qkd} is executed between neighbouring nodes
    and messages between non-neighbouring nodes are forwarded using a relay protocol.
    Although these networks are being deployed worldwide,
    no protocol exists which provides provable guarantees of integrity against manipulation from both external adversaries and corrupted intermediates.
    In this work, we present the first protocol that provably provides
    both confidentiality and integrity.
    Our protocol combines an existing cryptographic technique, \gls{amd} codes, with
    multi-path relaying over trusted repeater networks.
    This protocol achieves \gls{itsy} against the detection of manipulation,
    which we prove formally through a sequence of games.
\end{abstract}
\ifIACRCC
\begin{textabstract}
    Quantum Key Distribution (QKD) allows secure communication without relying on computational assumptions,
    but can currently only be deployed over relatively short distances due to hardware constraints.
    To extend QKD over long distances, networks of trusted repeater nodes can be used,
    wherein QKD is executed between neighbouring nodes
    and messages between non-neighbouring nodes are forwarded using a relay protocol.
    Although these networks are being deployed worldwide,
    no protocol exists which provides provable guarantees of integrity against manipulation from both external adversaries and corrupted intermediates.
    In this work, we present the first protocol that provably provides
    both confidentiality and integrity.
    Our protocol combines an existing cryptographic technique, Algebraic Manipulation Detection (AMD) codes, with
    multi-path relaying over trusted repeater networks.
    This protocol achieves Information-Theoretic Security (ITS) against the detection of manipulation,
    which we prove formally through a sequence of games.
\end{textabstract}
\fi
\section{Introduction}\label{sec:introduction}

Parties that can exchange quantum information can use \acrfull{qkd}
to expand a small shared secret into an arbitrarily long secret~\cite{bb84}.
The effectiveness of such quantum channels, both in terms of distance and rate of transmission, is currently restricted by hardware limitations.
While \emph{quantum repeaters} are able to lift this restriction
and establish long-distance end-to-end quantum channels,
their realisation is currently restricted to laboratory demonstrations
and cannot yet be deployed at scale in a quantum network. 

An alternative method for establishing keys over a long distance is to use trusted repeater networks~\cite{SPD+10}.
In such networks, \gls{qkd} is executed between neighbouring nodes of the network,
and messages between non-neighbouring nodes are then established via a \emph{relay} protocol.
In the simplest relay protocol, intermediate neighbours decrypt and re-encrypt:
if Alice wants to send a message to Bob,
she can encrypt and send it to Ingrid (an intermediate node) using a \gls{qkd}-generated key.
Ingrid can decrypt to recover the message and then re-encrypt and send it to Bob.
If more intermediate nodes are required to reach the distance,
each node can act as Ingrid and re-encrypt the key towards the next node.
Relaying does not provide end-to-end security,
since intermediate nodes must be trusted to forget and securely delete
the secret information after relaying.

To reduce the level of trust required of intermediate nodes,
parties can use a multi-path relaying protocol.
Assume for example that Alice and Bob do not share a direct \gls{qkd} link,
but both have a \gls{qkd} link with both Ingrid and Inez.
Alice generates two random secrets $S_1$ and $S_2$,
and relays both to Bob: she relays $S_1$ via Ingrid and $S_2$ via Inez.
Both Alice and Bob then output $S_1 \oplus S_2$ as their shared secret.
In order to learn the secret output, the adversary would have to corrupt both Ingrid and Inez.

This generalises to a network with $n$ disjoint paths,
each consisting of one or more intermediate nodes, shown in \cref{fig:graph}.
Alice splits her secret into $n$ shares using a secret-sharing scheme
and relays one share per path
from which Bob can recover the secret.

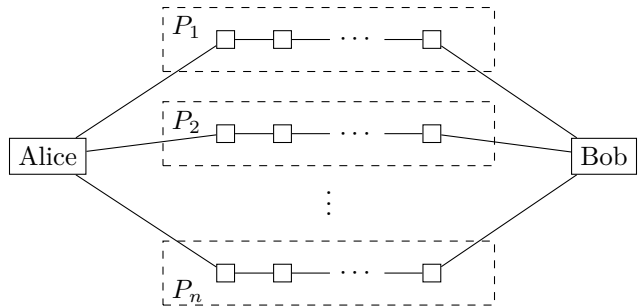
\begin{figure}
    \centering
    \begin{tikzpicture}[node distance=5mm]
    \node[draw] (n11) {};
    \node[draw, right=of n11] (n12) {};
    \node[right=of n12] (n13) {$\cdots$};
    \node[draw, right=of n13] (n14) {};
    \node[draw, dashed, inner xsep=7mm, inner ysep=3mm, anchor=north west, fit=(n11) (n14)] (p1) {};
    \node[anchor=north west] at (p1.north west) {$P_1$};

    \node[draw, below=10mm of n11] (n21) {};
    \node[draw, right=of n21] (n22) {};
    \node[right=of n22] (n23) {$\cdots$};
    \node[draw, right=of n23] (n24) {};
    \node[draw, dashed, inner xsep=7mm, inner ysep=3mm, anchor=north west, fit=(n21) (n24)] (p2) {};
    \node[anchor=north west] at (p2.north west) {$P_2$};

    \node[below=0mm of p2] {\vdots};

    \node[draw, below=16mm of n21] (nl1) {};
    \node[draw, right=of nl1] (nl2) {};
    \node[right=of nl2] (nl3) {$\cdots$};
    \node[draw, right=of nl3] (nl4) {};
    \node[draw, dashed, inner xsep=7mm, inner ysep=3mm, anchor=north west, fit=(nl1) (nl4)] (pl) {};
    \node[anchor=south west, yshift=-0.8mm] at (pl.south west) {$P_n$};

    \node[draw, left=10mm of $(p1.north west)!0.5!(pl.south west)$] (alice) {Alice};
    \node[draw, right=10mm of $(p1.north east)!0.5!(pl.south east)$] (bob) {Bob};

    \draw (alice) -- (n11) -- (n12) -- (n13) -- (n14) -- (bob);
    \draw (alice) -- (n21) -- (n22) -- (n23) -- (n24) -- (bob);
    \draw (alice) -- (nl1) -- (nl2) -- (nl3) -- (nl4) -- (bob);
\end{tikzpicture}
    \caption{Alice is connected to Bob via $
n$ disjoint paths.
      Each path $P_i$ consists of at most $\ell$ nodes,
      and no node may appear in more than one path.
    }%
    \label{fig:graph}
\end{figure}

\paragraph{Our contributions}
Note that if Bob's output differs from Alice's, then at least one share or message must have been tampered with.
In this case, it is desirable that Bob can detect tampering and reject the protocol run.
This property is called integrity; we claim that no existing solutions provide \emph{provable} integrity against such manipulation.
In this work, we present the first protocol which provides provable integrity against manipulation by computationally unbounded in trusted-repeater \gls{qkd} networks.
Our protocol, listed in \cref{alg:protocol} and visualised in \cref{fig:protocol_simple}, is conceptually simple.
First, Alice encodes her secret
using an existing construction from AMD codes (see \cref{ssec:amd}).
She then uses secret sharing (see \cref{ssec:secret-sharing})
and finally she relays each share over its own path,
using the \gls{qkd} generated keys as a \gls{otp}.
As we aim for \gls{itsy}, meaning that security must hold against
computationally unbounded adversaries,
we assume an upper bound on the number of paths that the adversary can corrupt. 
We formally prove security by adapting the existing game-based model~\cite{BR07}
to the context of trusted repeater networks,
and we show that the overhead (in number of additional \gls{qkd} bits consumed) is optimal,
and we provide an implementation to demonstrate our methods are efficient.

\subsection{Related Work}\label{sec:related_work}
Various solutions have been proposed which aim to achieve integrity in the trusted repeater context, most notably the SECOQC protocol~\cite{SPD+10}.
We summarise these solutions and
address the issues we found.

\subsubsection{SECOQC}\label{sec:secoqc}
As part of the SECOQC project,
Salvail, Peev, Diamanti, Alléaume, Lütkenhaus and Länger~\cite{SPD+10}
analysed the security of trusted repeater networks.
They proposed to use a secret-sharing scheme to exchange a secret confidentially,
followed by an interactive protocol to achieve integrity.
In their protocol, Alice locally computes random parity checks
over part of the secret-sharing scheme secret
and subsequently uses the remaining bits of the secret-sharing scheme output as a \gls{mac} key
to authenticate the transmission of the parity check equations.
The authors claim that their protocol enables Bob not only to \emph{detect} cheating
but also to \emph{identify} the cheater.

The issue with this approach is that \glspl{mac} only guarantee security under the assumption that the communicating parties already share an identical key.
This constitutes circular reasoning and thus the protocol does not yield provable security.
While we do not provide an explicit attack against cheater detection,
we do demonstrate in \cref{sec:disproof} that this allows an adversary to
break the cheater identification of the protocol.
The attack performs algebraic manipulation
to the \gls{mac} key, causing Bob to misidentify the malicious parties.

\subsubsection{Other solutions}
Beals, Hynes and Sanders proposed the \emph{stranger authentication protocol}:
an interactive protocol for checking the equality of a secret-sharing scheme output
which does not assume pre-shared keys~\cite{bs08,BHS09}.
In their protocol, Alice and Bob have secret-sharing scheme outputs
$S = S_1 \| S_2 \| S_3$ and
$S' = S_1' \| S_2' \| S_3'$, respectively.
Alice generates a random number $r$ and transmits $(r, H(S_3)) \oplus S_1$,
and Bob replies with $H(r) \oplus S_2'$.
The function $H$ is described as a hash function,
but its requirements are only specified informally
and it is unclear whether they can be satisfied 
against a computationally unbounded adversary.
The authors remark that their protocol assumes $S$ and $S'$ are unknown to the adversary,
whereas we observe that our method ensures integrity even if the adversary can choose $S$.

Another notable solution is presented by Rass and K\"onig~\cite{RK12}.
While the primary focus of their work is different from ours
as they address security against adversaries capable of rerouting
shares across the network,
the aspect relevant to our work is their method for detecting share manipulation.
They suggest that an ``obvious quickfix''
is to attach a checksum to the dealt secret.
Their proposed protocol employs a one-time \gls{mac} tag as a checksum,
however the authors neither specify a \gls{mac} key
nor clarify what an appropriate key would be.
They assume no pre-shared keys, and we observe that deriving the key from the dealt secret would lead to the same problematic circular reasoning observed in \cref{sec:secoqc}.

\paragraph{Outline}
In \cref{sec:bckgd}, we provide the necessary background for this work.
In \cref{sec:prot},
we provide our protocol (\cref{alg:protocol}) which uses an existing construction from cryptographic literature. We argue that the protocol achieves the desired security properties efficiently. We conclude with future directions in \cref{sec:concluding_section}.

\section{Background}\label{sec:bckgd}
In this section, we provide the necessary background on secret sharing, \gls{amd} codes and trusted repeater networks.
We will use the symbol $\bot$ to indicate a missing value,
with addition defined such that $\bot + x = x + \bot = \bot$ for all $x$.

\subsection{Secret Sharing}\label{ssec:secret-sharing}

A secret-sharing scheme allows a \emph{dealer} to split their secret into $n$ shares,
such that the secret can be completely recovered from a \emph{qualified} subset of shares,
but an \emph{unqualified} subset of shares reveals no information about the secret~\cite{shamir79}.
The subsets that are qualified are defined by an \emph{access structure}, which is a property of the specific scheme in use.
We adapt the definition of a secret-sharing scheme from other literature~\cite{cdfpw08}.

\begin{definition}[Secret-Sharing Scheme]
A secret-sharing scheme is a pair of functions $(\Share, \Rec)$.
Let $\mathcal{S}$ denote the secret space and $\mathcal{Y}^n$ the space of $n$ shares.
Given a secret $S \in \mathcal{S}$,
the probabilistic function $\Share$ splits $S$ into a vector
$\vect{S} = (S_1, \dots, S_n) \in \mathcal{Y}^n$,
called the \emph{shares} of $S$.
The function $\Rec$
recovers $S' \in \mathcal{S} \cup \{\bot\}$
from a vector of shares $\vect{S'} = (S_1', \dots, S_n') \in (\mathcal{Y} \cup \{\bot \})^n$.
Shares may be omitted, in which case $S_i' = \bot$ for some index $i$, and recovery may fail or reject, in which case $S' = \bot$.

Let $B \subseteq \{1, \dots, n\}$ be any \emph{qualified} set.
Then correctness requires that 
\[
\Rec(\vect{S'}) = S \quad \text{ for } \quad 
S_i' = \begin{cases}
    S_i & i \in B \\
    \bot & i \not\in B.\\
\end{cases}\]
Perfect confidentiality requires that for any \emph{unqualified} subset
$A \subset \{1, \dots, n\}$,
the shares $\{S_i\}_{i \in A}$ reveal no information about $S$.
\end{definition}

While the confidentiality requirement is presented somewhat informally in the above definition, it has a natural formalisation through the concept of indistinguishability~\cite{BR07}.

One of the simplest secret-sharing schemes is additive secret sharing, which distributes a secret as a sum of random shares.
In an additive group such that $S \in \mathcal{G}$,
$\Share(S)$ generates $(n-1)$ shares as uniformly random group elements
$S_i \getr \mathcal{G}$ for $1 \leq i < n$
and sets the last share to value $S_{n} = S - \sum_{i=1}^{n-1} S_i$.
$\Rec$ simply adds the shares: $S' = \sum S_i'$.
One important property of this scheme is that it is linear,
meaning that 
$\Rec(\vect{S} + \vect{S'})
= \Rec(\vect{S}) + \Rec(\vect{S'})$,
where vector addition is defined element-wise.\footnote{Formally, linearity holds only with respect to non-missing shares, which is sufficient for our purposes.}
In the additive secret sharing scheme, only the complete set of shares constitutes a qualified set.
A notable special case operates over the group of bitstrings under bitwise addition modulo 2, i.e.\ the bitwise exclusive-or ($\oplus$) operation.

Our method can easily be generalised to other linear secret-sharing schemes
with more flexible access structures,
such as Shamir's scheme~\cite{shamir79}.

\subsubsection{Robust secret sharing}
Robust secret sharing protects against an adversary that
aims to alter the recovery output ($S' \not\in \set{S, \bot}$)
by changing one or more of the shares ($S_i' \not\in \set{S_i, \bot}$).
Many definitions of robust secret sharing exist in the literature~\cite{MS81,tw89,ok96,BR07},
expressing subtly different security guarantees.
Some schemes merely want to \emph{detect} any tampering,
while others want to be able to \emph{correct},
i.e.\ recover $S' = S$ even when some $S_i' \neq S_i$.
Some schemes also aim to \emph{identify} whom the cheating shareholders are.
\emph{Strong} robustness protects when sharing any secret $S$,
while \emph{weak} robustness assumes that $S$ was sampled uniformly random and is unknown to the adversary~\cite{cdfpw08}.
Some schemes assume that recovery is done by an honest shareholder
(who knows that at least their own share is correct),
while others consider recovery by an external party.
For concrete security,
the distinction between \emph{dynamic} and \emph{static}
corruption is important~\cite{BR07}.

In general, robust secret sharing adds some redundancy to the shares,
increasing their size beyond that of the shared secret.
A strongly robust scheme will 
require $2\log1/\delta$ bits
to detect tampering with probability at least $\delta$~\cite{cdfpw08}.

The scheme we consider achieves detection,
no identification, external recovery,
and (a variant of) strong robustness,
all in the context of dynamic corruption.
We require a variant of robustness where the adversary
can choose some shares and induce a shift on uncorrupted shares,
which we introduce in \cref{sec:prot}, \cref{defn:shift-robust}.
In \cref{ssec:formal-int}, we adapt 
an existing game-based formalisation of robust secret sharing accordingly~\cite{BR07}.

\subsubsection{Weak robust secret sharing}\label{sec:weak_robust_ss}
Secret sharing can be used directly on secret messages,
or they can be used to establish a shared key between Alice and Bob.
In the latter case, weak robust secret sharing suffices,
which can provide the same security guarantees
with smaller overhead in the size of the shares.
Weak $\delta$-robustness was found to require growing the shares by $\log 1/\delta$ bits~\cite{CDV94}, an improvement of a factor two over strong robustness.
External difference families lead to weakly secure AMD codes with minimal tag size~\cite{cdfpw08}.

We summarise the two strands of literature in which weak robust secret sharing is implemented using \emph{difference sets}.
First, attempts had previously been made to implement Shamir's secret-sharing scheme in a way that protects against malicious shareholders, using elements of a field $F$~\cite{tw89} and a small subset  of legitimate values $S \subset F$.
Ogata and Kurosawa observed that the size of $S$ can be reduced by the use of difference sets~\cite{ok96}.
Ogata and Kurosawa proved that their construction is optimal for the size of shares and that the difference sets defined by Singer~\cite{singer38} suffice for this construction.
They also find useful parameter restrictions, which are that both $q$ and $q^2+q+1$ must be prime.
We note that this scheme is proven secure.
Although the scheme is optimal in terms of share size, it is not certain that there exists an efficient implementation of this scheme.
Specifically, we note that while it might be efficient to sample random points on a line in finite projective plane $PG(2,q)$ (as in Singer's work), finding the corresponding value $i$ in the difference set involves solving a discrete logarithm.
The search for an efficient implementation of the scheme remains open.

Our protocol does not employ weak robust secret sharing,
but we note that it is a possible avenue for improvement when sharing keys instead
of arbitrary messages.
In \Cref{sec:concluding_section}, we will reflect on the use of external difference families for weak robust secret sharing in trusted repeater \gls{qkd} networks.

\subsection{\texorpdfstring{\acrlong{amd}}{Algebraic Manipulation Detection (AMD)}}\label{ssec:amd}
In 2008, Cramer, Dodis, Fehr, Padró and Wichs provided an encoding of data that provides \acrlong{amd}~\cite{cdfpw08}.
Their code maps values to a group
in such a way that any nontrivial shift of an encoded value will be detected
with high probability by the decoder.

\begin{definition}[\acrfull{amd} code~\cite{cdfpw08}]%
\label{defn:amd-code}
    A $\delta$-\gls{amd} code is a pair of functions $(\Enc,\allowbreak\Dec)$:
    a probabilistic encoding map $\Enc: \mathcal{S} \to \mathcal{G}$, from some set $\mathcal{S}$ to a group $\mathcal{G}$,
    and a decoding map
    $\Dec: \mathcal{G} \to \mathcal{S} \cup \{\bot\}$
    such that
    $\Dec(\Enc(s)) = s$
    with probability 1 for all $s\in\mathcal{S}$.
    The security requirement is that
    for any $s \in \mathcal{S}$
    and any $\Delta \in \mathcal{G}$
    it holds that
    \[
    \Pr[\Dec(\Enc(s) + \Delta) \not\in \{s, \bot\}] \leq \delta.
    \]
\end{definition}

The authors also present the following construction and prove its security.
\begin{construction}[\cite{cdfpw08}]\label{const:encode}
Let $\mathbb{F}$ be a field of size $q$ and characteristic $p$,
and let $d$ be an integer such that $d+2$ is not divisible by $p$.
To encode a value $s = (s_1, \dots, s_{n}) \in \mathbb{F}^d$,
define $\Enc: \mathbb{F}^d \to \mathbb{F}^d \times \mathbb{F} \times \mathbb{F}$ given by $\Enc(s) \coloneq (s, x, f(x, s))$,
where $x \getr \mathbb{F}$ is sampled uniformly random,
and $f(x, s) \coloneq x^{d+2} + \sum_{i=1}^{d} s_i x^i$.
Decoding is naturally given by
\[
    \Dec(s', x', \sigma') = \begin{cases}
        s' & \text{if } f(x', s') = \sigma' \\
        \bot & \text{otherwise}.
    \end{cases}
\]
\end{construction}

The authors present several applications
that make use of the 
\gls{amd} code, including
the robust secret-sharing scheme $(\Share^*, \Rec^*)$ defined below.
\begin{construction}[\cite{cdfpw08}]\label{constr:robust-share}
    Let $(\Share, \Rec)$ be a linear secret-sharing scheme
    and let $(\Enc, \Dec)$ be an \gls{amd} code,
    then the robust secret-sharing scheme $(\Share^*, \Rec^*)$
    is defined by $\Share^*(s) = \Share(\Enc(s))$
    and $\Rec^*(\vect{S_i'}) = \Dec(\Rec(\vect{S_i'}))$.
\end{construction}
The authors prove that
\ifIACRCC
\autoref{constr:robust-share}
\else
\cref{constr:robust-share}
\fi
is a strongly robust secret-sharing scheme,
formalising robustness as the detection of tampering by an honest shareholder.

\subsection{Trusted Repeater Networks}\label{sec:trusted-repeater-networks}

A trusted repeater network is a graph $G = (V, E)$ where the vertices $v \in V$ (or nodes)
correspond to the protocol participants,
and edges $(v, v') \in E$ indicate which parties share a key.
Alice and Bob are nodes in the network;
other nodes represent trusted repeaters.
A path $P_i \subseteq E$ between two nodes is a sequence of edges that connect them.
Two paths are disjoint if none of their edges connect to the same vertex,
except for the first and last vertex.
We will only consider paths from Alice to Bob that are pairwise disjoint,
which allows for the following notation (see also \cref{fig:single-path-relay}).
We will provide values with two indices:
the first index $i$ refers to the $i$-th path,
while the second index $j$ refers to the $j$-th edge on that path.
For example, Alice shares a key $\qij{i}{1}{2}$ with her neighbour on path $P_i$.

\begin{figure}[ht]
  \centering
  \begin{tikzpicture}[node distance = 5mm and 26mm]
    \node (a) [draw] {Alice};
    \node (p1) [above right=of a, xshift=-20mm] {};
    \node (pn) [below right=of a, xshift=-20mm] {};
    \node (r1) [draw, right=of a] { };
    \node (r2) [draw, right=of r1] { };
    \node (b) [draw, right=of r2] {Bob};
    \node (p1l) [above left=of b, xshift=20mm] {};
    \node (pnl) [below left=of b, xshift=20mm] {};
    \node[draw, dashed, inner xsep=7mm, inner ysep=3mm, anchor=north west, fit=(a) (r1) (r2) (b)] (p2) {};
    \node[anchor=north west] at (p2.north west) {$P_i$};
    
    \draw[-, dashed] (a) -- (p1);
    \draw[-, dashed] (a) -- (pn);
    \draw[-, dashed] (p1l) -- (b);
    \draw[-, dashed] (pnl) -- (b);
    \draw[->] (a) --
        node [above] {$\cij{i}{1}{2} = m + \qij{i}{1}{2}$}
        node [below, gray] {$j=1$}
        (r1);
    \draw[->] (r1) --
        node [above] {$\cij{i}{2}{3} = m + \qij{i}{2}{3}$}
        node [below, gray] {$j=2$}
        (r2);
    \draw[->] (r2) --
        node [above] {$\cij{i}{3}{4} = m + \qij{i}{3}{4}$}
        node [below, gray] {$j=3$}
        (b);

    \node (ax) [below=of a] {};
    \node (bx) [below=of b] {};
    \draw [decorate,decoration={brace,amplitude=5pt,mirror}]
        (ax) -- (bx) node[midway,yshift=-10pt]{$\ell=3$};    
  \end{tikzpicture}
  \caption{Alice relays message $m$ to Bob on a single path $P_i$ of length $\ell=3$.
    Encryption on the $j$-th edge is done with an \Gls{otp} using key $\qij{i}{j}{j+1}$,
    resulting in ciphertext $\cij{i}{j}{j+1}$.}
  \label{fig:single-path-relay}
\end{figure}
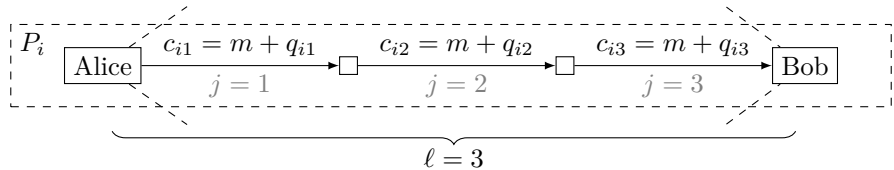

We consider only a simple relaying protocol,
in which each node receives an encrypted message from a neighbouring node, 
then decrypts the message and re-encrypts it 
before forwarding it to the next node along the designated path.
We will consider \gls{otp} encryption
in some additive group.
For example, in \cref{fig:single-path-relay},
the first trusted repeater will receive $\cij{i}{1}{}$ from Alice,
and will relay
$
    \cij{i}{2}{} = \cij{i}{1}{} - \qij{i}{1}{} + \qij{i}{2}{}
$
to the next trusted repeater.

If the adversary corrupts a node, they obtain all secrets held by that node.
A node is considered honest if it has not been corrupted;
Alice and Bob are assumed to be honest parties.
Honest parties securely delete keys after use.
Nodes that are not directly connected (i.e.\ non-neighbouring nodes)
do not share pre-established keys.
All communication between nodes is observable by the adversary,
who has full control over the network,
which includes altering, dropping, inserting, misdelivering, reordering and replaying messages.

Multipath protocols decrease the level of trust required of intermediate nodes.
The idea is as follows:
assume there are $n$ pairwise disjoint paths from Alice to Bob.
Alice acts as the dealer and splits her secret in $n$ shares,
then uses relaying to transmit each share over its own path
so that Bob can recover the secret
when he receives all shares.
The protocol provides confidentiality against any unqualified adversary,
defined as any adversary whose corruption of nodes does not reveal
a qualified set of shares, which is in turn defined by the secret sharing scheme.
For the additive sharing scheme,
this means any adversary that leaves all nodes uncorrupted on at least one path.
We present a multipath protocol that also provides integrity against any unqualified adversary.

\subsubsection{Single QKD Links}
While our results apply to trusted repeater networks in general,
our focus is on \gls{qkd} networks,
where current technological limitations prevent end-to-end quantum channels,
and trusted repeater networks are being deployed to overcome these distance limits.
In this context, the keys shared between neighbours are \gls{qkd}-generated keys.
In order to capture the scenario where an adversary can gain quantum side information (information about the quantum system via, for instance, entangled states), we use the security definition by Renner~\cite[Criterion~(2.6)]{Renner05}.

\begin{definition}[statistical privacy against quantum side information]\label{def:statprivacyqsideinfo}
  A classical-quantum state 
  \[
    \rho_{XE} = \sum_{x} P_X(x) |{x}\rangle \langle {x}| \otimes \rho_E^x,
  \]
  is said to be \(\varepsilon\)-private against quantum side information if
  \begin{equation}\label{eq:trace_dist}
      \frac{1}{2} \left\| \rho_{XE} - \rho_U \otimes \rho_E \right\|_1 \leq \varepsilon,
  \end{equation}
  where
  $\rho_{XE}$ is the joint state of the classical variable $X$ and the adversary's quantum system $E$,
  $\rho_U = \sum_x \frac{1}{|\mathcal{X}|} |{x}\rangle \langle {x}|$ is the uniform distribution over the support $\mathcal{X}$ of $X$,
  $\rho_E = \sum_x P_X(x) \rho_E^x$ is the reduced state of the adversary’s system and
  $\| \cdot \|_1$ denotes the trace norm.
\end{definition}
A distribution that is statistically $\varepsilon$-private is statistically indistinguishable from the uniform distribution by any unbounded adversary.

In \gls{qkd}, the adversary is usually given complete control over the quantum channel. 
Therefore, the adversary can always cause the \gls{qkd} protocol to abort and not produce a key. 
In our scenario, repeated aborts between two neighbouring \gls{qkd} nodes $v, v'$ in the network would result in the loss of the edge $(v,v')$ in the trusted repeater network. 
One could model this power of the adversary by considering networks with probabilistic or adversarial edges. 
However, for simplicity, we assume a \emph{static} network topology. 
In other words, we assume for every edge $(v,v') \in E$, the adversary does not induce a protocol abort. 
Similarly, a \gls{qkd} protocol run could be incorrect, i.e.\ adversary interference or inherent protocol imperfections could cause the generated key to be different for the two parties, therefore rendering this link invalid. 
Again, we assume for simplicity that this does not occur for any edge $(v,v') \in E$. 
Therefore, the imperfection of a \emph{single} \gls{qkd} link is captured by \cref{def:statprivacyqsideinfo} above.

\subsubsection{QKD Networks}
While \cref{def:statprivacyqsideinfo} above captures secrecy of a single \gls{qkd} link, our scenario requires a network of \gls{qkd} links. 
In this scenario, we consider a central adversary
with control over all (quantum) communication in the network.
We therefore generalise \cref{def:statprivacyqsideinfo} in a natural way.

\begin{definition}[collective statistical privacy against quantum side information]\label{def:collectivestatprivacyqsideinfo}
  For $k$ classical variables $X^k = (X_1, X_2, \ldots X_k)$, we call the state
  \[
    \rho_{X^k E} = \sum_{x^k} P_{X^k}(x^k) |{x^k}\rangle \langle {x^k}| \otimes \rho_E^{x^k},
  \]
  \(\varepsilon\)-private against quantum side information if
  \begin{equation}\label{eq:trace_dist_k}
      \frac{1}{2} \left\| \rho_{X^k E} - \rho_U \otimes \rho_E \right\|_1 \leq \varepsilon,
  \end{equation}
  where
  $\rho_{X^k E}$ is the joint state of the classical variables $X^k$ and the adversary's quantum system $E$,
  $\rho_U = \sum_{x^k} \frac{1}{|\mathcal{X}|^k} |{x^k}\rangle \langle {x^k}|$ is the uniform distribution over the support $\mathcal{X}^k$ of $X^k$,
  $\rho_E = \sum_{x^k} P_{X^k}(x^k) \rho_E^{x^k}$ is the reduced state of the adversary's system and
  $\| \cdot \|_1$ denotes the trace norm.
\end{definition}

\begin{lemma} \label{lem:nQKDsystems}
In a \gls{qkd} network with $k$ individual \gls{qkd} links that are $\varepsilon$-private according to \cref{def:statprivacyqsideinfo},
the overall state is $k \varepsilon$-private according to \cref{def:collectivestatprivacyqsideinfo}.
\end{lemma}
\begin{proof}[Proof sketch]
Intuitively, the claim follows from the universal composability of \gls{qkd} keys. 
Concretely, the randomness of the keys generated in the individual \gls{qkd} runs is generated by honest players, and hence the keys are independent of each other.
Therefore, a global eavesdropper does not have the ability to correlate the generated keys.
The claimed overall privacy bound of $k \varepsilon$ then follows by a union bound from the $\varepsilon$-privacy of the individual keys.
\end{proof}

\section{Using \texorpdfstring{\gls{amd}}{AMD} Codes in Trusted Repeater Networks}\label{sec:prot}

Our protocol is straightforward: Alice shares a secret using
\ifIACRCC
\autoref{constr:robust-share}
\else
\cref{constr:robust-share}
\fi
and relays the shares over the trusted repeater network.
The protocol is
defined in \cref{alg:protocol} and
visualised in \cref{fig:protocol_simple}.
We use notation $\cij{i}{j}{}$ for the ciphertext message
sent on the $j$-th edge of the $i$-th path,
$\cpij{i}{j}{}$ for the delivered ciphertext
(possibly changed by the adversary)
and $\qij{i}{j}{}$ for the corresponding key.

\begin{algorithm}[ht]
  \caption{Protocol for achieving confidentiality and integrity in trusted relay networks.}\label{alg:protocol}%
  \begin{algorithmic}[1]
  \Procedure{Alice}{$s$}
    \State $\vect{S} \gets \Share^*(s)$
    \Comment{$\vect{S} = (S_1, \dots, S_{n}) = \Share(\Enc(s))$}
    \For{$1 \leq i \leq n$}
      \State $c_i' \gets S_i + \qij{i}{1}{2}$
      \Comment{Encrypt. $\vect{c} = (\cij{1}{1}{2}, \dots \cij{n}{1}{2})$}
      \State \textbf{delete} $\qij{i}{1}{}$
    \EndFor
    \State \Return $\vect{c}$
    \Comment{Ciphertext $c_i$ is sent to path $P_i$}
  \EndProcedure
  \Procedure{Relay$_{ij}$}{$\cpij{i}{(j-1)}{}$}
    \Comment{The $j$-th node on path $i$ receives ciphertext $\cpij{i}{(j-1)}{j}$}
    \State $\cij{i}{j}{} \gets \cpij{i}{(j-1)}{} - \qij{i}{(j-1)}{} + \qij{i}{j}{}$
    \Comment{Re-encrypt}
    \State \textbf{delete} $\qij{i}{(j-1)}{}$ \textbf{and} $\qij{i}{j}{}$
    \State \Return $\cij{i}{j}{}$
    \Comment {Ciphertext $\cij{i}{j}{}$ is sent to the $(j+1)$-th node on path $i$}
  \EndProcedure
  \Procedure{Bob}{$\vect{c'}$}
    \Comment{Ciphertext $c_i'$ is received from path $P_i$}
    \For{$1 \leq i \leq n$}
      \State $S_i' \gets c' - \qij{i}{\ell}{}$
      \Comment{Decrypt. $\vect{S'} = (S'_1, \dots, S'_{n})$}
      \State \textbf{delete} $\qij{i}{\ell}{}$
    \EndFor
    \State $s' \gets \Rec^*(\vect{S'})$
    \Comment{Bob outputs $s' = \Dec(\Rec(\vect{S}))$}
  \EndProcedure
  \end{algorithmic}
\end{algorithm}

\begin{figure}[ht]
    \centering
    \begin{tikzpicture}[node distance=10mm and 24mm]
    \node[draw] (alice) {Alice};
    \node[above=0mm of alice, align=left, xshift=-18mm] (ainit) 
    {
        $\vect{S} = \Share^*(s)$\\
        $\forall i: \cij{i}{1}{2} = S_i + \qij{i}{1}{2}$
    };
    \node[draw, dashed, right=of alice] (p1) {$P_2$};
    \node[
        draw,
        right=of p1,
    ] (bob) {Bob};
    \node[draw, dashed, above=of p1] (p0) {$P_1$};
    \node[above right=0mm of p0]
      {$\cij{1}{2}{3} = \cpij{1}{1}{2} - \qij{1}{1}{2} + \qij{1}{2}{3}$};
    \node[draw, dashed, below=of p1] (pl) {$P_3$};
    
    \node[below right=0mm of bob, align=left,xshift=-2mm] (bfinal) 
    {
        $\forall i: S_i' = \cpij{i}{2}{3} - \qij{i}{2}{3}$\\
        $s' = \Rec^*(\vect{S'})$
    };

    \draw[->] (alice) --
        node [above, sloped, very near start] {$\cij{1}{1}{2}$}
        node [above, sloped, very near end] {$\cpij{1}{1}{2}$}
        (p0);
    \draw[->] (p0) --
        node [above, sloped, very near start] {$\cij{1}{2}{3}$}
        node [above, sloped, very near end] {$\cpij{1}{2}{3}$}
        (bob);
    \draw[->] (alice) --
        node [above, sloped, near start,xshift=1mm] {$\cij{2}{1}{2}$}
        node [above, sloped, very near end] {$\cpij{2}{1}{2}$}
        (p1);
    \draw[->] (p1) --
        node [above, sloped, very near start] {$\cij{2}{2}{3}$}
        node [above, sloped, near end, xshift=-1mm] {$\cpij{2}{2}{3}$}
        (bob);
    \draw[->] (alice) --
        node [below, sloped, very near start] {$\cij{3}{1}{2}$}
        node [below, sloped, very near end] {$\cpij{3}{1}{2}$}
        (pl);
    \draw[->] (pl) --
        node [below, sloped, very near start] {$\cij{3}{2}{3}$}
        node [below, sloped, very near end] {$\cpij{3}{2}{3}$}
        (bob);
\end{tikzpicture}
    \caption{
        Protocol for transmitting a secret $s$
        over the encrypted channels of the trusted repeater network,
        visualised here for $n=3$ and $\ell=2$.
        Alice and Bob output $s$ and $s'$, respectively,
        where we note that $s' = \bot$ if decoding fails,
        indicating that Bob rejects.
    }
    \label{fig:protocol_simple}
\end{figure}
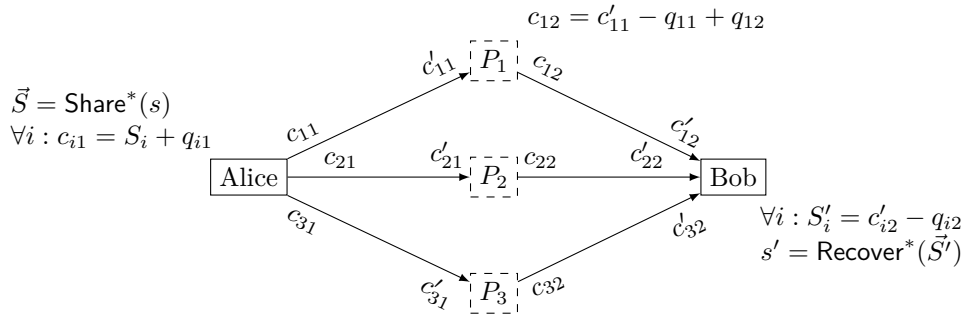

\subsection{Security analysis}\label{ssec:security-analysis}

For this work we make strong assumptions on the trusted repeater network:
the network $G$ is fixed,
the endpoints are fixed,
and even the keys used between neighbours are fixed.
We reflect on these assumptions in \cref{sec:assumptions}.

We first consider the confidentiality of the protocol; this reasoning is formalised in \cref{ssec:formal-conf}.
\begin{theorem}[Confidentiality (informal)]\label{thm:conf}
  Fix a trusted repeater network with $n$ disjoint paths,
  each of length at most $\ell$,
  where all neighbours share a key that is $\varepsilon$-private against quantum side information.
  Let $(\Share, \Rec)$ be a perfectly confidential secret-sharing scheme.
  Then the protocol of \cref{alg:protocol}
  provides $n \ell \varepsilon$-confidentiality
  against all unqualified adversaries.
\end{theorem}

\paragraph{Proof sketch.}
Each encryption with an $\varepsilon$-private key leaks at most
$\varepsilon$ information about the encrypted share.
There are $n \ell$ total encryptions;
thus, by \cref{lem:nQKDsystems},
the total information leaked is $n \ell \varepsilon$.
Given our assumption that the adversary corrupts only an unqualified set of paths or shares,
the perfect confidentiality of the secret-sharing scheme guarantees no further leakage. \qed
\vspace{1em}

We note that the network width $n$ occurs in the above bound
because the adversary may corrupt nodes dynamically.
In the case of static corruption,
the total leakage can be bounded by summing over only the fixed, uncorrupted path(s).

Next we consider integrity of the protocol.
\ifIACRCC
\autoref{constr:robust-share},
\else
\Cref{constr:robust-share},
\fi
$(\Share^*, \Rec^*)$,
is known to be robust
when a shareholder recovers the secret~\cite[Theorem 3]{cdfpw08}.
In contrast, recovery is external in the trusted repeater network setting:
Bob cannot assume receipt of at least one honest share.
Given that \gls{otp} encryption is malleable,
the adversary can make a blind change to the honest shares.
We therefore generalise robustness to account for this adversarial capability.

\begin{definition}[Shift-robustness]\label{defn:shift-robust}
  A secret-sharing scheme $(\Share, \Rec)$ is \emph{$\delta$-shift-robust}
  if for any secret $s$, and for any computationally unbounded adversary $\mathcal{A}$
  that corrupts an unqualified set of paths
  $A \subset \{1, \dots, n\}$,
  we have the following:
  Let $\vect{S} = \Share(s)$,
  let $a_i = \mathcal{A}(i, \{S_a\}_{a \in A}) \in \mathcal{G} \cup \set{\bot}$ be the $i$-th adversarially chosen value
  and let $\vect{S'}$ be a vector
  such that, for $1 \leq i \leq n$,
  \[
    S'_i = \begin{cases}
        a_i & \text{if } i \in A \\
        S_i + a_i & \text{if } i \not\in A
      \end{cases}
  \]
then $\Pr[\Rec(\vect{S'}) \not\in \{s, \bot\}] \leq \delta$.
\end{definition}

We now prove that
\ifIACRCC
\autoref{constr:robust-share}
\else
\cref{constr:robust-share}
\fi
is $\delta$-shift robust. 
The proof is very similar to the proof for robust secret sharing~\cite{cdfpw08}.

\begin{lemma}\label{lem:shift-sss}
    Let $(\Enc, \Dec)$ be a $\delta$-\gls{amd} code,
    and $(\Share, \Rec)$ be a perfect linear secret-sharing scheme code,
    then the secret-sharing scheme $(\Share^*, \Rec^*)$
    defined in
    \ifIACRCC
    \autoref{constr:robust-share}
    \else
    \cref{constr:robust-share}
    \fi
    is a $\delta$-shift robust secret-sharing scheme.
\end{lemma}

\begin{proof}
Let $S = \Share^*(s)$ and
let $\vect{S'}$ be a vector meeting the requirements of \cref{defn:shift-robust}.
Let $\vect{\Delta} = \vect{S'} - \vect{S}$.
By the linearity of the secret-sharing scheme we have
\begin{align*}
  \Pr[\Rec^*(\vect{S'}) \not\in \set{s, \bot}]
  &= \Pr[\Dec(\Rec(\vect{S}) + \Rec(\vect{\Delta})) \not\in \set{s, \bot}] \\
  &= \Pr[\Dec(\Enc(s) + \Delta)) \not\in \set{s, \bot}]
\end{align*}
where $\Delta = \Rec(\vect{\Delta})$ is determined by the adversarial strategy.
By the perfect confidentiality of the secret-sharing scheme,
$\Enc(s)$ is unknown to the adversary (except for a-priori knowledge of $s$), 
so $\Delta$ is independent of the randomness used in $\Enc(s)$.
The lemma then follows immediately from the definition of \gls{amd} codes (\cref{defn:amd-code}).
\end{proof}

The protocol of \cref{alg:protocol}
provides integrity against any unqualified adversary,
meaning they corrupted only an unqualified set of paths.
We argue integrity informally below and formalise this reasoning later in \cref{ssec:formal-int}.

\begin{theorem}[Integrity (informal)]\label{thm:integrity}
  Fix a trusted repeater network with $n$ disjoint paths,
  each of length at most $\ell$,
  where all neighbours share a key that is $\varepsilon$-private against quantum side information.
  Let $(\Share, \Rec)$ be a perfectly confidential
  and $\delta$-shift robust secret-sharing scheme.
  Then the protocol of \cref{alg:protocol} provides
  $(n \ell \varepsilon + \delta)$-integrity
  against all unqualified adversaries.
\end{theorem}

\paragraph{Proof sketch}
The proof proceeds through a sequence of games.
In the first gamehop, we replace all \gls{qkd} keys with uniform group elements;
by \cref{lem:nQKDsystems}, the resulting game is indistinguishable,
except with probability at most $n \ell \varepsilon$.
In the second gamehop,
we reduce to the security of the shift-robust secret-sharing scheme.
By \cref{lem:shift-sss}, the winning probability of the last game is $\delta$,
so we add that term.
We now outline the reduction for the second gamehop,
assuming $n=3$ and $\ell=2$ as in \cref{fig:protocol_simple}
for simplicity.
The full version of this proof sketch in \cref{ssec:formal-int} generalizes this to any $n$ and $\ell$.

Given any adversary $\mathcal{A}$ that breaks the integrity of the relay protocol
with perfect keys,
we construct an adversary $\mathcal{B}$
against the $\delta$-shift robustness of the secret-sharing scheme
by simulating the relay protocol of $\mathcal{A}$.
We only consider one intermediate per path to avoid notational clutter;
the formal proof extends to paths of arbitrary length.

Let $\vect{c} = (\cij{1}{1}{2}, \cij{2}{1}{2}, \cij{3}{1}{2})$ denote the ciphertexts sent by Alice;
$\mathcal{B}$ simulates these ciphertexts by sending a uniformly random vector.
When $\mathcal{A}$ corrupts path $P_i$,
$\mathcal{B}$ corrupts the $i$-th shareholder to get $S_i$
and derives $\qij{i}{1}{2} = \cij{i}{1}{2} - S_i$.
The other key is sampled uniformly random ($\qij{i}{2}{3} \getr \mathcal{G}$).
When $\mathcal{A}$ queries a relay on a corrupted path,
the relay returns the computed key.
On uncorrupted paths, the relay returns a uniformly random ciphertext, 
which is equivalent to relaying with a random key $\qij{i}{2}{3}$.
Finally, $\mathcal{A}$ outputs $\vect{c'} = (\cpij{1}{2}{3}, \cpij{2}{2}{3},\allowbreak \cpij{3}{2}{3})$
and from this $\mathcal{B}$ derives the vector $\vect{a} = (a_1, a_2, a_3)$
as follows:
if a path is corrupted, then $\mathcal{B}$ can simply decrypt to obtain
$a_i = \cpij{i}{2}{} - \qij{i}{2}{}$, as Bob would do in the relay protocol.
For uncorrupted paths, $\mathcal{B}$ sets $a_i$ to the sum of differences between received and sent ciphertexts:
\[
a_i = \cpij{i}{2}{3} - \cij{i}{2}{3} + \cpij{i}{1}{2} - \cij{i}{1}{2}.
\]
Recovery decrypts with $\qij{i}{2}{3}$ to get
$S_i' = \cpij{i}{2}{3} - \qij{i}{2}{3}$.
From
\begin{align*}
\cij{i}{1}{} &= S_i + \qij{i}{1}{} \\
\cij{i}{2}{} &= \cpij{i}{1}{} - \qij{i}{1}{} + \qij{i}{2}{}
\end{align*}
we get the required value
\begin{align*}
S_i + a_i
&= S_i + \cpij{i}{2}{} - \cij{i}{2}{} + \cpij{i}{1}{} - \cij{i}{1}{2} \\
&= S_i + \cpij{i}{2}{3} - (\cpij{i}{1}{2} - \qij{i}{1}{} + \qij{i}{2}{}) + \cpij{i}{1}{} - (S_i + \qij{i}{1}{2}) \\
&= \cpij{i}{2}{3} - \qij{i}{2}{3}
= S_i'.
\end{align*}

Note that if no message was relayed on path $P_i$ and thus $\cij{i}{2}{3}$ and $\cpij{i}{1}{2}$ are not defined when $\mathcal{A}$ outputs $\vect{c'}$,
then $\mathcal{B}$ can sample $a_i$ uniformly from random.
This is equivalent to decryption with a uniformly random key.
\qed
\subsection{Properties of Protocol}
We outline the advantages of our protocol and examine the validity of notable assumptions.
\subsubsection{Advantages}
We highlight three advantages of the proposed protocol.
First, we improve significantly on all existing solutions.
Above all, our protocol is proven to be secure,
unlike any other existing protocol.
Additionally, we note that the SECOQC protocol~\cite{SPD+10} is an
interactive protocol with $3n\ell$ authenticated messages
between network neighbours,
while in our protocol is non-interactive
(Alice only sends and Bob only receives)
and requires $n\ell$ non-authenticated messages between neighbours.

Second, our protocol yields an efficient implementation.
For a field $\mathbb{F}$ of size $q$,
the encoding of
\ifIACRCC
\autoref{const:encode}
\else
\cref{const:encode}
\fi
maps a value $s \in \mathbb{F}^d$
to values in $\mathbb{F}^d \times \mathbb{F} \times \mathbb{F}$,
and additive secret sharing has no overhead.
For strong robustness (i.e. for relaying arbitrary messages), we achieve the optimal overhead of $2\log1/\delta$ bits~\cite{cdfpw08}.
To demonstrate the efficiency and simplicity,
we provide a concrete implementation in \cref{sec:code}.

Third, we note that our protocol works with any linear secret sharing scheme, allowing users to choose their trade-off between efficiency and security.
If $n$-out-of-$n$ secret sharing is desired, then the simple additive scheme can be used.
If it is beneficial to allow recovery from fewer than $n$ shares, for instance if network edges might run out of \gls{qkd} key material, then Shamir's $t$-out-of-$n$ scheme with $t < n$ may be a preferable choice of secret sharing scheme.
It should be noted that reducing the size of qualified sets also reduces the number of paths an adversary would need to corrupt in order to recover the secret.

\subsubsection{Assumptions}\label{sec:assumptions}
The first notable assumption relates to node authentication.
In \gls{ake} parties typically identify their peers by knowledge of some secret key:
in asymmetric \gls{ake} this key is the private component of a public-key pair,
while in symmetric \gls{ake} this is the shared key.
In trusted repeater networks, however, such an end-to-end key is lacking.
In trusted repeater networks, a node's identity 
can at best be determined by its location within the network.
Consequently, authentication in such a network relies heavily on the trustworthiness of neighbouring nodes and their connections.
In models where the adversary has some control
over the selection of the endpoints and/or the relaying paths,
authentication will depend on
whether Bob can (cryptographically) verify if the claimed endpoint and paths are correct.
By the strong assumption we made that the network is static,
the assumption of node authentication is achieved.
This is only reasonable if it can be ensured either by the physical properties of the network or by employing an additional (cryptographic) protocol that provides this guarantee
in a dynamic network that is only partially trustworthy.

The second notable assumption is, as detailed in \cref{ssec:security-analysis}, that elements of this model are fixed.
This inability to represent dynamic scenarios is a significant limitation.
In particular, it assumes that both parties and communication paths remain static,
and that the keys used for relaying are known in advance,
assumptions that restrict practicality severely.
Real-world implementations may require greater flexibility, such as the ability to execute the protocol multiple times between different parties.

Finally, we reflect on the choice to add integrity on top of an existing network, which required the above assumptions about the network.
Our protocol ensures that any modifications to messages travelling over this pre-authenticated network, whether by external parties or untrustworthy nodes, are detected, which we describe as integrity.
Existing literature has referred to this concept as authenticity~\cite{cdfpw08,BHS09} or robustness~\cite{RK12}.
Although the term robustness is more consistent with the existing literature on secret-sharing schemes, we avoid using it due to its potential ambiguity in this work.
We do not propose that adding integrity on top of an authenticated network
is the sole method of securing a trusted repeater network; 
the combination of confidentiality, integrity and authenticity 
within a single protocol may offer a more effective approach 
to reducing the assumptions on the network.

\section{Future Directions}\label{sec:concluding_section}
We outline three future directions which emerge from this work.
First, as we recall from \cref{sec:weak_robust_ss} that an external difference family would give a weakly secure AMD code with minimal tag size, and as weak robustness is sufficient when relaying keys, a natural future direction is to search for an external difference family which satisfies the parameter constraints detailed in \cref{sec:weak_robust_ss}.
Second, as we recall from \cref{sec:disproof} that the SECOQC protocol aimed unsuccessfully at being able to identify cheaters, a natural future direction would be to create a protocol with this property.
Third, as noted in \cref{sec:assumptions}, our assumption that the network is static severely restricts its practicality.
Rather than proposing refinements to our protocol within the current static model, we conclude with a broader open problem: the development of a trusted repeater model capable of representing dynamic networks.
\printbibliography
\appendix
\section{Attack on the SECOQC Protocol}\label{sec:disproof}
We summarise the SECOQC protocol illustrated in \cref{fig:protocol_secowc}.
The protocol is executed after a secret-sharing scheme,
and the core concept is the generation of a random parity check matrix $\Lambda$, 
with the subsequent computation of parity bits $r$ over the secret.
A two-time \gls{mac} key $\kappa$ is taken from the secret-sharing scheme output,
and tags are attached to subsequent messages to prevent adversarial modification.\footnote{
  A two-time \gls{mac} key could be the concatenation of two one-time \gls{mac} keys
  or, in the case of Wegman-Carter \glspl{mac},
  a function index and two one-time pads.
}
The authors claim that this protocol can both detect and \emph{identify} the malicious paths
by inspecting from which path they receive invalid parity bits or tags.

As previously noted, the issue with this approach is that \glspl{mac} only guarantee security
when used with shared \emph{identical} keys $\kappa = \kappa'$, and using the equality of values as keys in a
MAC scheme constitutes circular reasoning and does not yield provable security.
We now present a successful attack
against the claimed identification property of the protocol.
With this attack, the adversary is able to cause honest paths to be
misidentified as malicious,
while their own corrupted path is misidentified as honest.

\begin{figure}
    \centering
    \begin{tikzpicture}[node distance=10mm and 24mm]
    \node[draw] (alice) {Alice};
    \node[align=left, above=0mm of alice, xshift=-15mm] (alabel) {
        $\kappa \| s := S$\\
        $\Lambda \leftarrow \{0,1\}^{m \times (n+m)}$\\
        $r := \Lambda s$\\
        $T := \MAC_\kappa(\Lambda \| r)$
    };
    \node[draw, dashed, right=of alice] (p1) {$P_2$};
    \node[
        draw,
        right=of p1,
    ] (bob) {Bob};
    \node[align=left, below=0mm of bob, xshift=23mm] (blabel) {
        $\kappa' \| s' := S'$\\
        accept iff $\exists i$:\\
        \quad$\Lambda_i s' = r_i$ and\\
        \quad$T_{i} = \MAC_{\kappa'}(\Lambda_i \| r_i)$
    };
    \node[draw, dashed, above=of p1] (p0) {$P_1$};
    \node[below=3mm of p1] (dots) {$\vdots$};
    \node[draw, dashed, below=15mm of p1] (pl) {$P_{n}$};

    \draw[->] (alice) -- node [above, sloped] {$(\Lambda, r, T)$} (p0);
    \draw[->] (p0) -- node [above, sloped] {$(\Lambda_1, r_1, T_{1})$} (bob);
    \draw[->] (alice) -- node [above, sloped, pos=.6] {$(\Lambda, r, T)$} (p1);
    \draw[->] (p1) -- node[above, sloped, pos=.4] {$(\Lambda_2, r_2, T_{2})$} (bob);
    \draw[->] (alice) -- node [below, sloped] {$(\Lambda, r, T)$} (pl);
    \draw[->] (pl) -- node[below, sloped] {$(\Lambda_{n}, r_{n}, T_{n})$} (bob);
\end{tikzpicture}
    \caption{
      SECOQC integrity protocol~\cite{cdfpw08}.
      This depicts step (2), the full protocol runs as follows:
      (1) Alice and Bob run a (non-robust) additive secret-sharing scheme
      to get $S$ and $S'$ respectively.
      (2) They run the depicted protocol
      so Bob can verify whether $S' = S$.
      (3) Bob replies with $(b, \MAC_\kappa(b))$,
      where $b$ is the accept/reject bit.
      (4) Alice and Bob deterministically select some bits from $s$
      to remove in order
      to correct for the information leaked via $(\Lambda, r)$.
    }%
    \label{fig:protocol_secowc}
\end{figure}
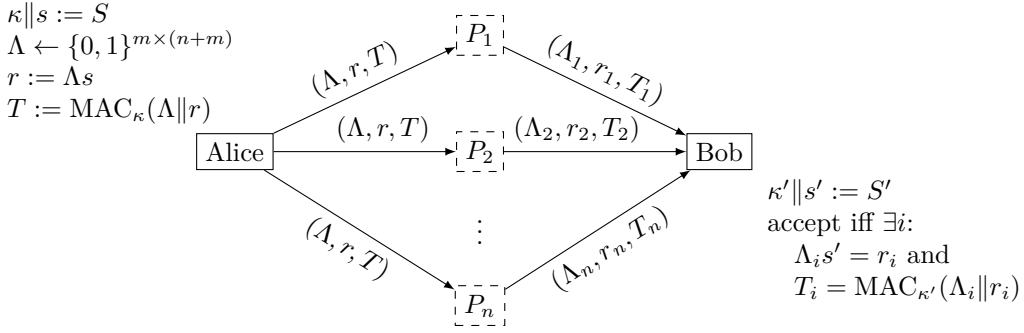

We let the \gls{mac} be implemented with a Wegman-Carter construction~\cite{WC81}.
Let $\mathcal{F}$ be an almost XOR-universal class of functions.
Split the \gls{mac} key in three parts $\kappa_1 \| \kappa_2 \| \kappa_3 := \kappa$,
where $\kappa_1$ identifies a function $f_{\kappa_1} \in \mathcal{F}$
and $\kappa_2$ and $\kappa_3$ will act as one-time pads.
Alice computes her tag as
\begin{equation}\label{eq:wcmac}
  T := \MAC_{\kappa}(\Lambda\|r) = f_{\kappa_1}(\Lambda\|r) \oplus \kappa_2.
\end{equation}

Consider an adversary against the SECOQC protocol
that controls a node in path $P_{n}$.
Let this adversary select some algebraic shift $\Delta_2$
and manipulate step (1) so that Alice gets share
\begin{equation}
  S_{n} = \kappa_{1n} \| \kappa_{2n} \| \kappa_{3n} \| s_{n}
\end{equation}
and Bob gets share 
\begin{equation}
  S_{n}' = \kappa_{1n} \| (\kappa_{2n} \oplus \Delta_2) \| \kappa_{3n} \| s_{n}.
\end{equation}
When combined with the other shares, 
Alice computes \gls{mac} key $\kappa = \kappa_1 \| \kappa_2 \| \kappa_3$ and
Bob computes $\kappa' = \kappa_1 \| (\kappa_2 \oplus \Delta_2) \| \kappa_3$.
In step (2), Alice computes $T$ with \cref{eq:wcmac}
and sends $T$ on every path.
On path $P_n$, the adversary adds $\Delta_2$
to the received tag, so that Bob receives
\begin{equation}
  T_{n} = T \oplus \Delta_2,
\end{equation}
while other paths forward the tag unmodified: $T_{i} = T$ for $i < n$.\footnote{
  If the adversary controls other paths, they can also deliver $T_{n}$ on those paths.
}
When Bob verifies $T_{n}$ with key $\kappa'$,
he accepts the tag because
\begin{equation}
  T_{n} = T \oplus \Delta_2 = f_{\kappa_1}(\Lambda \| r) \oplus \kappa_2 \oplus \Delta_2 = \MAC_{\kappa'}(\Lambda \| r).
\end{equation}
Since $s = s'$ and the adversary did not change $\Lambda$ and $r$,
Bob also finds that the parity information is correct, so he accepts.
At the same time, Bob finds that $T_{i} \neq T_{n}$ for all other paths,
so he wrongly concludes that those paths must be dishonest.

\section{Formalisation of Security Definitions and Proofs}\label{sec:formal}

We adapt Bellare and Rogaway's formalisation of a secret-sharing scheme~\cite{BR07}, in which confidentiality and robustness are modelled through distinct games.
We do not modify their definition of confidentiality apart from minor notational adjustments.
For robustness, we modify the existing model by 
incorporating explicit cheater detection and
allowing the adversary to shift uncorrupted shares arbitrarily.
We assume that the number of shares and the access structure are fixed parameters of the secret-sharing scheme.
We will consider confidentiality and integrity through two separate games, $\oIndRelay{}$ and $\oForgeRelay{}$.

We formally define the security 
using established game-based notions,
where the probability that the adversary breaks security
is defined as the probability of them winning the associated game.
A hybrid proof is a sequence of hops between games such that there are only minor changes between successive games,
so that we can bound the difference between the winning probabilities of subsequent games.
The initial game describes the protocol and the final game models a game with a known bound on the winning probability.

One aspect not covered by the separate modelling of confidentiality and integrity
is the potential interplay between the two properties.
For instance, if the adversary can see the accept/reject bit,
this could inadvertently reveal information about the secret.
We leave the development of an appropriate extension of the model which addresses this limitation for future work.

\subsection{Confidentiality}\label{ssec:formal-conf}

The formal indistinguishability game $\oIndSSS{}$
for a secret-sharing scheme $\Pi_s$ is given in \cref{alg:ind-sss}.
An adversary $\mathcal{A}$ provides two secrets (subroutine $\mathcal{A}_{0}$),
the game selects one at random and executes \Share{}.
The adversary can corrupt any unqualified set of shares,
tracked by the game in a set $T$ (\cref{line:T}),
and then guesses which share the game picked (subroutine $\mathcal{A}_1$).
The game outputs one (denoted with the indicator function $\ind{\cdot}$)
if the adversary is unqualified and guessed correct.
A secret-sharing scheme $\Pi_s$ has \emph{perfect} confidentiality if
for all $\mathcal{A}$ we have
\begin{equation}\label{eq:ind-sss}
\Pr[\oIndSSS(\Pi_s, \mathcal{A}) = 1] = 1/2,
\end{equation}
in other words, the adversary cannot do better than guessing randomly.

\begin{algorithm}
  \caption{Indistinguishability game for secret-sharing scheme $\Pi_s = (\Share, \Rec)$ \cite{BR07}.
  }
  \label{alg:ind-sss}
  \begin{algorithmic}[1]
  \oGame{\oIndSSS}{$\Pi_s, \mathcal{A}$}
    \State $b \getr \set{0,1}$
      \Comment{Challenge bit $b$}
    \State $s_0, s_1 \gets \mathcal{A}_0()$
      \Comment{Adversary chooses secrets $s_0, s_1$}
    \State $\vect{S} \getr \Share(s_b)$
      \Comment $\vect{S} = (S_1, \dots, S_{n})$
    \State $b' \gets \mathcal{A}_1^{\oCorruptS(\cdot)}()$
      \Comment{Adversary guesses $b'$}
    \State \Return $\ind{ b = b' \textbf{ and } T \text{ is unqualified} }$\label{line:T}
  \oEndGame
  \oQuery{\oCorruptS}{i}
    \State $T \gets T \cup \set{i}$
      \Comment{Initially $T = \emptyset$}
    \State \Return $S_i$
  \oEndQuery
  \end{algorithmic}
\end{algorithm}

\begin{algorithm}
  \caption{Indistinguishability game for multipath relay protocol $\Pi_r$
  using the secret-sharing scheme $(\Share, \Rec)$
  over a trusted repeater network.}
  \label{alg:ind-relay}
  \begin{algorithmic}[1]
  \oGame{\oIndRelay}{$\Pi_r, \mathcal{A}$}
    \State $b \getr \set{0,1}$
    \State $s_0, s_1 \gets \mathcal{A}_0()$
    \State $\vect{S} \getr \Share(s_b)$
      \Comment $\vect{S} = \big(S_1, \dots, S_{n} \big)$
    \For{$1 \leq i \leq n \textbf{ and } 1 \leq j \leq \ell$}
      \State $\qij{i}{j}{j+1} \getr \mathrm{QKD}$\label{line:sampleq}
        \Comment{\gls{qkd} generated keys as per \cref{def:statprivacyqsideinfo}}
    \EndFor
    \For{$1 \leq i \leq n$}
      \State $\cij{i}{1}{2} = S_i + \qij{i}{1}{2}$
      \Comment $\vect{c} = \big(\cij{1}{1}{2}, \dots, \cij{n}{1}{2} \big)$
    \EndFor
    \State $b' \gets \mathcal{A}_1^{\oRelay(\cdot, \cdot, \cdot), \oCorruptR(\cdot, \cdot)}(\vect{c})$\Comment{$\vect{c'} = \big( \cpij{1}{\ell}{},\dots,\cpij{n}{\ell}{} \big)$}
    \State \Return $\ind{ b = b' \textbf{ and } T \text{ is unqualified} }$
  \oEndGame
  \oQuery{\oRelay}{$i, j, \cpij{i}{(j-1)}{j}$}
    \If{$(i, j) \in R$}\label{line:R}
      \State \Return $\bot$
        \Comment{Relay is allowed at most once to model key deletion}
    \EndIf
    \State $R \gets R \cup \set{(i,j)}$
      \Comment{Initially $R = \emptyset$}
    \State \Return $\cpij{i}{(j-1)}{j} - \qij{i}{(j-1)}{j} + \qij{i}{j}{j+1}$
      \Comment{Re-encrypt}\label{line:reencrypt}
  \oEndQuery
  \oQuery{\oCorruptR}{$i, j$}
    \If{$(i, j) \in R$}
      \State \Return $\bot$
        \Comment{Keys were already deleted after relaying}
    \EndIf
    \State $T \gets T \cup \set{i}$
    \State \Return $(\qij{i}{(j-1)}{j}, \qij{i}{j}{j+1})$
  \oEndQuery
  \end{algorithmic}
\end{algorithm}

In \cref{alg:ind-relay}, we define the indistinguishability game for trusted repeater networks.
The main differences with $\oIndSSS{}$ are that
corruption reveals the relay keys and that the adversary 
gets full control over the relaying by deciding which (possibly forged)
ciphertext to deliver to which node and in which order.
For simplicity, the game is defined for a protocol $\Pi_r$ that uses secret sharing
and simple relaying,
but it can be generalised by extracting the protocol-specific code into subroutines.
Note that $R$ (\cref{line:R}) is a set containing information
about which nodes have already relayed;
$(i,j) \in R$ means that key $\qij{i}{j}{j+1}$ was already deleted
because $\oRelay(i,j,\cdot)$ was called before.

\begin{theorem}[Confidentiality]\label{thm:relay}
Let $G$ be a
trusted-repeater network with $n$ paths, each of length at most $\ell$,
let the \gls{qkd} keys be $\varepsilon$-private against quantum side information according to \cref{def:statprivacyqsideinfo},
and let $\Pi_r$ be the protocol of \cref{alg:protocol}.
Then, for all (computationally unbounded) adversaries $\mathcal{A}$:
\[
    \left| \Pr[\oIndRelay(\Pi_r, \mathcal{A})=1] - 1/2 \right| \leq n \ell \varepsilon.
\]
\end{theorem}

\begin{proof}
Recall the definition of $\varepsilon$-privacy against quantum side information (\cref{def:statprivacyqsideinfo}).
We give the proof as two game hops.
We define game $G_1$ as equal to \oIndRelay{} (\cref{alg:ind-relay}),
except that we replace \cref{line:sampleq} with the following so that $G_1$ samples all keys uniformly random:
$\qij{i}{j}{j+1} \getr \mathcal{G}$. 
Since individual QKD keys are $\varepsilon$-private according to \cref{def:statprivacyqsideinfo},
and there are $n\ell$ keys in total,
it follows from \cref{lem:nQKDsystems} that
\begin{equation}\label{eq:ind-sss-g1}
    \left| \Pr[\oIndRelay(\Pi_r, \mathcal{A}) = 1] - \Pr[G_1 = 1] \right| \leq n \ell \varepsilon.
\end{equation}

We now show that $G_1$ is secure by giving a reduction $\mathcal{B}$
to the $\oIndSSS{}$ game,
detailed in \cref{alg:ind-reduction}.
While the reduction involves detailed steps, its structure is conceptually straightforward.
Before corruption, relaying returns random ciphertexts
to simulate re-encryption with a random key (\cref{line:relay-before}).
Upon corruption of path $P_i$,
$\mathcal{B}$ corrupts in the secret sharing game to get the share (\cref{line:corrupt})
and computes all the keys in the path.
For any ciphertext seen by the adversary $\mathcal{A}$ of game $G_1$,
the key is chosen so as to produce a value consistent with the adversary's view (\cref{line:consistent_value,line:computes_keys}),
for example $\qij{i}{1}{2} = \cij{i}{1}{2} - S_i$.
Keys for unseen ciphertexts are sampled uniformly random (\cref{line:samples_keys_random}).
After corruption, all relays are computed with the derived keys (\cref{line:relay-after}).
Finally, $\mathcal{B}_1$ outputs whatever $\mathcal{A}_1$ outputs (\cref{line:A_output}).
Since $\mathcal{B}$ wins \oIndSSS{} exactly when $\mathcal{A}$ wins $G_1$,
we have
\begin{equation}\label{eq:ind-sss-g2}
    \Pr[G_1 = 1] = \Pr[\oIndSSS(\Pi_s, \mathcal{B}) = 1].
\end{equation}

\begin{algorithm}
  \caption{\oIndSSS{} adversary $\mathcal{B} = (\mathcal{B}_0, \mathcal{B}_1)$,
  using $G_1$ adversary $\mathcal{A} = (\mathcal{A}_0, \mathcal{A}_1)$.}
  \label{alg:ind-reduction}
  \begin{algorithmic}[1]
  \Procedure{$\mathcal{B}_0$}{\,}
    \State \Return $\mathcal{A}_0()$
  \EndProcedure
  \Procedure{$\mathcal{B}_1$}{\,}
    \State $\vect{c} \getr \mathcal{G}^n$
    \Comment{Initialise with random ciphertexts}
    \State \Return $\mathcal{A}_1^{\oRelay(\cdot, \cdot, \cdot), \oCorruptR(\cdot, \cdot)}(\vect{c})$\label{line:A_output}
  \EndProcedure
  \oQuery{\oRelay}{$i, j, \cpij{i}{(j-1)}{j}$}\label{line:red-oRelay}
    \If{$(i, j) \in R$}
      \State \Return $\bot$
    \EndIf
    \State $R \gets R \cup \set{(i, j)}$
    \If{$i \in T$}
      \State $\cij{i}{j}{j+1} \getr \mathcal{G}$\label{line:relay-before}
      \Comment{Return random ciphertext on uncorrupted paths}
    \Else
      \State $\cij{i}{j}{j+1} \gets \cpij{i}{(j-1)}{j} - \qij{i}{(j-1)}{j} + \qij{i}{j}{j+1}$\label{line:relay-after}
      \Comment{Re-encrypt with keys computed at corruption}
    \EndIf
    \State \Return $\cij{i}{j}{j+1}$
  \oEndQuery
  \oQuery{\oCorruptR}{$i, j$}\label{line:red-oCorrupt}
    \If{$(i, j) \in R$}
      \State \Return $\bot$
    \EndIf
    \If {$i \not\in T$}
      \State $T \gets T \cup \set{i}$
      \State $S_i \gets \oCorruptS(i)$\label{line:corrupt}
      \State $\qij{i}{1}{2} \gets \cij{i}{1}{2} - S_i$
      \Comment{Set to value consistent with adversary view (first step)}\label{line:consistent_value}
      \For{$2 \leq j' \leq \ell$}
        \If{$(i,j') \in R$}
          \Comment{Note that $\cij{i}{j'}{j'+1} \neq \bot$ and $\cpij{i}{(j'-1)}{j'} \neq \bot$}
          \State $\qij{i}{j'}{j'+1} \gets \cij{i}{j'}{j'+1} - \cpij{i}{(j'-1)}{j'} + \qij{i}{(j'-1)}{j'}$
          \Comment{Consistency with adversary view}\label{line:computes_keys}
        \Else
          \State $\qij{i}{j'}{j'+1} \getr \mathcal{G}$
          \Comment{Other keys sampled from random}\label{line:samples_keys_random}
        \EndIf
      \EndFor
    \EndIf
    \State \Return $(\qij{i}{(j-1)}{j}, \qij{i}{j}{j+1})$
  \oEndQuery
  \end{algorithmic}
\end{algorithm}

Since the secret-sharing scheme is perfectly secure (\cref{eq:ind-sss}),
we can combine \cref{eq:ind-sss,eq:ind-sss-g1,eq:ind-sss-g2}
to conclude
\[%
    |\Pr[\oIndRelay(\Pi_r, \mathcal{A}) =1] - 1/2 | \leq n \ell \varepsilon. \qedhere%
\]%
\end{proof}

\subsection{Integrity}\label{ssec:formal-int}

We define shift-robustness for a secret-sharing scheme where
shares $S_i$ are elements of some additive group $\mathcal{G}$.
Formally, this is modelled by the game \oShiftSSS{} defined in \cref{fig:shiftsss}.
An unqualified adversary can select the initial secret,
corrupt shares to learn and change them,
and shift uncorrupted shares by adding any chosen group element,
all with the goal of having the game accept a secret different than the one initially provided.
A secret-sharing scheme $\Pi_s = (\Share, \Rec)$ is said to be $\delta$-shift-robust
if for all adversaries $\mathcal{A}$ we have
\begin{equation}\label{eq:shift-robust}
  \Pr[\oShiftSSS(\Pi_s, \mathcal{A}) = 1] \leq \delta.
\end{equation}

\begin{algorithm}
  \caption{Shift-Robustness game for secret-sharing scheme $\Pi_s = (\Share, \Rec)$, adapted from \cite{BR07}.
  }
  \label{fig:shiftsss}
  \begin{algorithmic}[1]
  \oGame{\oShiftSSS}{$\Pi_s, \mathcal{A}$}
    \State $s \gets \mathcal{A}_0()$
      \Comment{Adversary chooses secret $s$}
    \State $\vect{S} \getr \Share(s)$
    \State $\vect{a} \gets \mathcal{A}_1^{\oCorruptS(\cdot)}()$
      \Comment{Adversary chooses shares and shifts: $\vect{a} = (a_1, \dots, a_{n})$}
    \For{$1 \leq i \leq n$}
      \If{$i \in T$}
        \State $S_i' \gets a_i$
          \Comment{Corrupted share replaced by $a_i$}
      \Else
        \State $S_i' \gets S_i + a_i$
          \Comment{Add $a_i$ to uncorrupted share}
      \EndIf
    \EndFor
    \State \Return $\ind{ \Rec(\vect{S'}) \not\in \set{s, \bot} \textbf{ and } T \text{ is unqualified} }$
      \Comment{$\vect{S'} = \big(S_1', \dots, S_{n}' \big)$}
  \oEndGame
  \oQuery{\oCorruptS}{i}
    \State $T = T \cup \set{i}$
    \State \Return $S_i$
  \oEndQuery
  \end{algorithmic}
\end{algorithm}

By \cref{lem:shift-sss}, we know that
\ifIACRCC
\autoref{constr:robust-share},
\else
\cref{constr:robust-share},
\fi
namely $(\Share^*, \Rec^*)$ instantiated with a $\delta$-robust \gls{amd}-encoding,
is $\delta$-shift-robust.

\begin{algorithm}
  \caption{Integrity game for relay protocol $\Pi_r$,
  using the secret-sharing scheme $(\Share, \Rec)$
  over a trusted repeater network.}
  \label{fig:int-relay}
  \begin{algorithmic}[1]
  \oGame{\oForgeRelay}{$\Pi_r, \mathcal{A}$}
    \State $s \gets \mathcal{A}_0()$
    \State $\vect{S} \gets \Share(s)$
    \For{$1 \leq i \leq n \textbf{ and } 1 \leq j \leq \ell$}
      \State $\qij{i}{j}{j+1} \getr \mathrm{QKD}$\label{line:replace_rand}
        \Comment{\gls{qkd} generated keys as per \cref{def:statprivacyqsideinfo}}
    \EndFor
    \For{$0 \leq i < n$}
      \State $\cij{i}{1}{2} = S_i + \qij{i}{1}{2}$
      \Comment{$\vect{c} = \big( \cij{1}{1}{},\dots,\cij{n}{1}{} \big)$}
    \EndFor
    \State $\vect{c'} \gets \mathcal{A}_1^{\oRelay(\cdot, \cdot, \cdot), \oCorruptR(\cdot, \cdot)}(\vect{c})$
      \Comment{$\vect{c'} = \bigl( \cpij{1}{\ell}{},\dots,\cpij{n}{\ell}{} \bigr)$}
    \For{$1 \leq i \leq n$}
      \State $S_i' \gets \cpij{i}{\ell}{} - \qij{i}{\ell}{}$
    \EndFor
    \State \Return $\ind{ \Rec(\vect{S}') \not\in \set{s, \bot} \textbf{ and } T \text{ is unqualified} }$
      \Comment{$\vect{S}' = \big(S_1', \dots, S_{n}' \big)$}
  \oEndGame
  \oQuery{\oRelay}{$i, j, \cpij{i}{(j-1)}{j}$}
    \If{$(i, j) \in R$}
      \State \Return $\bot$
        \Comment{Relay at most once, to model deletion of keys}
    \EndIf
    \State $R \gets R \cup \set{(i,j)}$
      \Comment{Initially $R = \emptyset$}
    \State \Return $\cij{i}{j}{(j+1)} \gets \cpij{i}{(j-1)}{j} - \qij{i}{(j-1)}{j} + \qij{i}{j}{(j+1)}$ 
  \oEndQuery
  \oQuery{\oCorruptR}{$i, j$}
    \If{$(i, j) \in R$}
      \State \Return $\bot$
        \Comment{Keys were already deleted after relaying}
    \EndIf
    \State $T \gets T \cup \set{i}$
    \State \Return $(\qij{i}{(j-1)}{j}, \qij{i}{j}{(j+1)})$
  \oEndQuery
  \end{algorithmic}
\end{algorithm}

In \cref{fig:int-relay} we define the game for modelling
the integrity of multipath relaying protocols in the trusted repeater network context.
Similar to the $\oIndRelay{}$ game, we give the adversary the power to corrupt nodes
and to control relaying. In addition, we let the adversary choose $\vect{c'}$,
which is then decrypted and recovered by the game.

\begin{theorem}[Integrity]\label{thm:gamehop}
Let $G$ be a
trusted-repeater network with $n$ paths, each of length at most $\ell$,
let the \gls{qkd} keys be $\varepsilon$-private against quantum side information according to \cref{def:statprivacyqsideinfo},
and let $\Pi_r$ be the protocol of \cref{alg:protocol}.
Then, for all (computationally unbounded) adversaries $\mathcal{A}$:
\[
    \Pr[\oForgeRelay(\Pi_r, \mathcal{A}) = 1] \leq n \ell \varepsilon + \delta.
\]
\end{theorem}

\begin{proof}
We define game $G_2$ as equal to \oForgeRelay{} (\cref{fig:int-relay}),
except that we replace \cref{line:replace_rand} with the following so that all \gls{qkd} keys are sampled as uniformly random group elements:
$\qij{i}{j}{j+1} \getr \mathcal{G}$. 
Since individual QKD keys are $\varepsilon$-private against quantum side information (\cref{def:statprivacyqsideinfo}) and there are $n\ell$ keys,
it follows from \cref{lem:nQKDsystems} that
\begin{equation}\label{eq:int-hop1}
  |\Pr[\oForgeRelay(\Pi_r, \mathcal{A}) = 1] - \Pr[G_2 = 1]| \leq n \ell \varepsilon.
\end{equation}

We now show that $G_2$ is secure by giving a reduction $\mathcal{B}$, detailed in \cref{alg:integrity-reduction}: $\mathcal{B}$ is a \oShiftSSS{} adversary and simulates game $G_2$ against black-box adversary $\mathcal{A}$.
Note that the inputs to $\mathcal{A}$ are identical to that of \cref{alg:ind-reduction}:
$\vect{c}$ is sampled uniformly random and
the simulation of oracles \oRelay{} and \oCorruptR{} are identical.
However, when $\mathcal{A}_1$ returns a vector
$\vect{c'}$
on \cref{line:c_vec}, $\mathcal{B}_1$ needs to translate this into a vector
$\vect{a}$
for the \oShiftSSS{} challenger (\cref{line:returns_a}),
depending on the status of relaying and corruption per path.
First, missing ciphertexts are translated into missing shifts (\cref{line:ai_not}).
Second, ciphertexts on corrupted paths can simply be decrypted (\cref{line:ai_decrypts}).
Third, incomplete relays on uncorrupted paths will be decrypted with
a unused uniformly random key in $G_2$, 
resulting in a uniformly random shift (\cref{line:not_fully_relayed}).
Finally, in a complete end-to-end relay,
$\mathcal{B}_1$ computes the total shift (\cref{line:computes_total_shift})
as the sum over the path of differences between delivered and sent ciphertexts:
$
    a_i = \sum_{j=1}^{\ell} \left( \cpij{i}{j}{(j+1)} - \cij{i}{j}{(j+1)} \right)
$.

\begin{algorithm}
  \caption{\oIndSSS{} adversary $\mathcal{B} = (\mathcal{B}_0, \mathcal{B}_1)$,
  using $G_2$ adversary $\mathcal{A} = (\mathcal{A}_0, \mathcal{A}_1)$.}
  \label{alg:integrity-reduction}
  \begin{algorithmic}[1]
  \Procedure{$\mathcal{B}_0$}{\,}
    \State \Return $\mathcal{A}_0()$
  \EndProcedure
  \Procedure{$\mathcal{B}_1$}{\,}
    \State $\vect{c} \getr \mathcal{G}^n$
    \Comment{$
      \vect{c} = \big(\cij{1}{1}{2}, \dots, \cij{n}{1}{2}\big)$}
    \State $\vect{c'} \gets \mathcal{A}_1^{\oRelay(\cdot, \cdot, \cdot), \oCorruptR(\cdot, \cdot)}(\vect{c})$ 
    \Comment{$
      \vect{c'} = \big(\cpij{1}{\ell}{}, \dots, \cpij{n}{\ell}{} \big)
      $}\label{line:c_vec}
    \For{$1 \leq i \leq n$}
      \If{$\cpij{i}{\ell}{} = \bot$}
        \Comment{No relay message delivered}
        \State $a_i \gets \bot$\label{line:ai_not}
      \ElsIf{$\qij{i}{\ell}{} \neq \bot$}
        \Comment{$\oCorruptR(i)$ was called}
        \State $a_i \gets \cpij{i}{\ell}{} - \qij{i}{\ell}{}$\label{line:ai_decrypts}
      \ElsIf{$\exists j : 1 \leq j < \ell \textbf{ and } \cpij{i}{j}{j+1} = \bot$}
        \Comment{Not fully relayed}
        \State $a_i \getr \mathcal{G}$\label{line:not_fully_relayed}
      \Else
        \Comment{Fully relayed}
        \State $a_i \gets \sum_{j=1}^{\ell} ( \cpij{i}{j}{j+1} - \cij{i}{j}{j+1} )$
        \label{line:computes_total_shift}
      \EndIf
    \EndFor
    \State \Return $\vect{a}$
      \Comment{$\vect{a} = (a_0, \dots, a_{n-1})$}\label{line:returns_a}
  \EndProcedure
  \oQuery{\oRelay}{$i, j, \cpij{i}{(j-1)}{j}$} \dots{}
    \Comment{See \cref{alg:ind-reduction}, \cref{line:red-oRelay}}
  \oEndQuery
  \oQuery{\oCorruptR}{$i, j$} \dots{}
    \Comment{See \cref{alg:ind-reduction}, \cref{line:red-oCorrupt}}
  \oEndQuery
  \end{algorithmic}
\end{algorithm}

For the full relay on an uncorrupted path in $G_2$,
we have $\cij{i}{1}{2} = S_i + \qij{i}{1}{2}$ and $\cij{i}{j}{(j+1)} = \cpij{i}{(j-1)}{j} - \qij{i}{(j-1)}{j} + \qij{i}{j}{(j+1)}$ for $2 \leq j \leq \ell$.
Therefore we see that the simulated shift $a_i$ indeed encodes the end-to-end difference
between the sent and received share:
\begin{align*}
    S_i + a_i
    &= S_i + \sum_{j=1}^{\ell} \bigl(\cpij{i}{j}{j+1} - \cij{i}{j}{j+1}\bigr) \\
    &= S_i + \cpij{i}{1}{2} - \cij{i}{1}{2}
       + \sum_{j=2}^{\ell} \bigl(\cpij{i}{j}{j+1} - \cij{i}{j}{j+1}\bigr) \\
    &= S_i + \cpij{i}{1}{2} - (S_i + \qij{i}{1}{2})
       + \sum_{j=2}^{\ell} \bigl(\cpij{i}{j}{j+1} - (\cpij{i}{(j-1)}{j} - \qij{i}{(j-1)}{j} + \qij{i}{j}{j+1})\bigr) \\
    &= \cpij{i}{\ell}{} - \qij{i}{\ell}{} = S_i',
\end{align*}
where the fourth equality follows since we have a telescoping sum
and the final equality depicts decryption in $G_2$.

Then $\mathcal{B}$ wins the \oShiftSSS{} game exactly when $\mathcal{A}$ would win game $G_2$.
From this we conclude
\begin{equation}\label{eq:int-hop2}
  \Pr[G_2 = 1] = \Pr[\oShiftSSS(\Pi_s, \mathcal{B}) = 1].
\end{equation}

By \cref{eq:shift-robust},
we know that this probability is bounded by $\delta$.
Combining \cref{eq:shift-robust,eq:int-hop1,eq:int-hop2}, we get
\[ \Pr[\oForgeRelay(\Pi_r, \mathcal{A}) = 1] \leq n \ell \varepsilon + \delta. \qedhere \]
\end{proof}
\clearpage
\section{Implementation}\label{sec:code}
The script in \cref{lst:impl} performs secret sharing using the additive scheme and \gls{amd} encoding.
It is provided with the intention of assisting engineers in the practical deployment of robust secret sharing within trusted repeater networks in \gls{qkd} systems.
\begin{lstlisting}[style=sage,
caption={SageMath proof of concept implementation of AMD encoded additive secret sharing.
See 
\anonurl{https://github.com/sebastianv89/AMD-secret-sharing}
for a version
including unit-tests and dynamic parameter selection.},
label=lst:impl]
"""WARNING: PROOF OF CONCEPT, DO NO USE FOR SECURING REAL DATA!"""

# System parameters (example)
d = 3
F = GF(2^86)

def share(n, secret):
    """Return n shares, shared via additive sharing."""
    shares = [[F.random_element() for _ in range(len(secret))]
              for _ in range(n-1)]
    shares.append([secret[i] - sum(share[i] for share in shares)
                   for i in range(len(secret))])
    return shares

def recover(shares):
    """Recover the secret from the additive shares."""
    return [sum(share[i] for share in shares)
            for i in range(len(shares[0]))]

def tag(x, s):
    """Compute f(x, s), using Horner's method."""
    res = x
    for coef in s:
        res = x*res + coef # IMPORTANT: multiplication in GF(q)
    return x*res

def amd_encode(s):
    """Return (s, x, f(x, s)): the systematic encoding of s."""
    x = F.random_element()
    return s + [x, tag(x, s)]

def amd_decode(encoded):
    """Return s or None: the decoding of (s, x, y)."""
    *s, x, y = encoded
    if tag(x, s) == y:
        return s
    return None

def amd_share(n, secret):
    """Return n shares, shared via AMD encoded additive sharing."""
    return share(n, amd_encode(secret))

def amd_recover(shares):
    """Recover the secret from the AMD encoded additive shares."""
    return amd_decode(recover(shares))
\end{lstlisting}
\end{document}